\documentclass[11pt]{amsart}
\usepackage[utf8]{inputenc}
\usepackage{url}
\newcommand{\term}[1]{\emph{#1}}
\newcommand{\F}{F}
\newcommand{\FX}{\F[X]}
\newcommand{\CZ}{S(f)}

\DeclareMathOperator{\rad}{rad}
\newtheorem{thm}{Theorem}
\newtheorem{lem}[thm]{Lemma}
\newtheorem{prop}[thm]{Proposition}
\newtheorem*{alg}{Algorithm}
\author{Przemysław Koprowski}
\address{Faculty of Mathematics\\
  University of Silesia\\
  ul. Bankowa 14\\ 
  PL-40-007 Katowice, Poland}
\email{pkoprowski@member.ams.org}
\title{Roots multiplicity without companion matrices}
\begin{document}
\begin{abstract}
We show a method for constructing a polynomial interpolating roots' multiplicities of another polynomial, that does not use companion matrices. This leads to a modification to Guersenzvaig--Szechtman square-free decomposition algorithm that is more efficient both in theory and in practice.
\end{abstract}
\maketitle

The problem of computing the square-free decomposition of a polynomial is well established in the realm of computational algebra. There are efficient algorithms for this task developed nearly half a century ago by: R.~Tobey \cite{Tobey67}, E.~Horowitz \cite{Horowitz69}, D.~Musser \cite{Musser71} and D.~Yun \cite{Yun76}. Recently, N.~Guersenzvaig and F.~Szechtman invented a completely new algorithm (see \cite{GS12}). They associate to a given polynomial~$f$ its \term{roots-multiplicity polynomial}~$M_f$. The authors presented a formula, for constructing~$M_f$, based on the companion matrix of the radical of~$f$. The aim of this note is to show that~$M_f$ can be obtained much more efficiently without using the companion matrix.

For readers convenience (and to some discomfort of the author), in this paper we use the same notation as N.~Guersenzvaig and F.~Szechtman did. Let~$\F$ be a fixed field of characteristic~$0$. Given a monic polynomial $f\in\FX$, the \term{square-free decomposition} of~$f$ is an expression
\begin{equation}\label{eq_sqfree_factorization}
f = P_1\cdot P_2^2\dotsm P_m^m,
\end{equation}
where each $P_k\in \FX$ for $k\in \{1, \dotsc, m\}$ is monic and square-free. Let $\CZ$ be the set of all the roots of~$f$ (in some algebraically closed field). For every root $\alpha\in \CZ$, denote by $m(\alpha)$ the multiplicity of~$\alpha$, so that
\[
f = \prod_{\alpha\in \CZ} (X-\alpha)^{m(\alpha)}.
\]
By Lagrange interpolation formula, there is a unique polynomial~$M_f$ of minimal degree such that $M_f(\alpha) = m(\alpha)$ for every root~$\alpha$ of~$f$. Once the polynomial~$M_f$ is found, the square free factors $P_1, \dotsc, P_m$ of~$f$ can be computed by the formula \cite[Eq.~(1.5)]{GS12}:
\[
P_k = \gcd(M_f - k, r)\qquad\text{for }k\in \{1, \dotsc, m\}.
\]
Thus, all we need is an efficient method for constructing~$M_f$. 

Guersenzvaig and Szechtman proposed the following procedure. The \term{radical} $r:=\rad f$ is a polynomial
\[
r = \prod_{\alpha\in \CZ}(X-\alpha) = \frac{f}{\gcd(f,f')}.
\]
The radical is square-free, hence its degree, denote it by $s:=\deg r$, equals the cardinality of $\CZ = S(r)$. Suppose that~$r$ has a form $r = r_0 + r_1X + \dotsb + r_{s-1}X^{s-1} + X^s$. Let
\[
C_r := 
\begin{pmatrix}
0 &        &   & -r_0\\
1 & \ddots &   & \vdots\\
  & \ddots & 0 & -r_{s-2}\\
  &        & 1 & -r_{s-1}
\end{pmatrix}
\]
be a companion matrix of~$r$ and set 
\[
P := \frac{f'}{\gcd(f,f')}.
\]
The fact that the radical is square-free implies that it is relatively prime to its derivative. Therefore, by B\'ezout identity, there are polynomials $g,h\in \FX$ such that
\begin{equation}\label{eq_xgcd}
r'\cdot g + r\cdot h = 1
\qquad\text{and}\qquad
\deg g < \deg r,\ \deg h < \deg r'.
\end{equation}
For any polynomial~$p$, $\deg p < s$, by $[p]$ denote a column vector of its coefficients, zero-appended to length~$s$, if needed.

\begin{thm}[{\cite[Theorem~2.1]{GS12}}]\label{thm_Mf_original}
With the above notation
\begin{equation}\label{eq_Mf_original}\tag{A}
[M_f] = P(C_r)\cdot [g].
\end{equation}
\end{thm}

We claim that~$M_f$ can be constructed more efficiently. To this end we need:

\begin{lem}\label{lem_p_xi}
For every root $\alpha\in \CZ$ one has
\[
P(\alpha) = m(\alpha)\cdot r'(\alpha).
\]
\end{lem}

\begin{proof}
Compute the derivative of $r$:
\[
r' = \sum_{\alpha\in \CZ} \prod_{\substack{\beta\in \CZ\\ \beta\neq \alpha}} (X-\beta).
\]
It follows that for every root $\alpha$ of~$f$ we have
\begin{equation}\label{eq_r'xi}
r'(\alpha) = \prod_{\substack{\beta\in \CZ\\ \beta\neq \alpha}}(\alpha - \beta).
\end{equation}
Next, compute the derivative of~$f$:
\[
f' = \sum_{\alpha\in\CZ} \Bigl( m(\alpha)\cdot (X-\alpha)^{m(\alpha)-1}\cdot 
\prod_{\substack{\beta\in\CZ\\ \beta\neq \alpha}} (X - \beta)^{m(\beta)}\Bigr).
\]
It follows that the polynomial~$P$ can be expressed in a form
\begin{align*}
P 
&= \frac{f'}{\gcd(f,f')}
= \frac{r\cdot f'}{f}\\
&= \frac{\sum\limits_{\alpha\in \CZ} \Bigl( m(\alpha)\cdot (X-\alpha)^{m(\alpha)-1}\cdot \prod\limits_{\beta\neq \alpha} (X-\beta)^{m(\beta)}\Bigr)}%
{\prod\limits_{\alpha\in \CZ} (X-\alpha)^{m(\alpha)-1}}\\
&= \sum_{\alpha\in \CZ} \Bigl( m(\alpha)\cdot \prod_{\substack{\beta\in \CZ\\ \beta\neq \alpha}}(X-\beta)\Bigr).
\end{align*}
Consequently, evaluating~$P$ at a root~$\alpha$, we have
\[
P(\alpha) = m(\alpha)\cdot \prod_{\substack{\beta\in \CZ\\ \beta\neq \alpha}} (\alpha - \beta)
\]
and the thesis follows from Eq.~\eqref{eq_r'xi}.
\end{proof}

\begin{prop}
The polynomial~$M_f$ is the remainder of the product $P\cdot g$ modulo~$r$:
\begin{equation}\label{eq_Mf_new}\tag{B}
M_f = \bigl( P\cdot g\bmod r\bigr).
\end{equation}
\end{prop}

\begin{proof}
Denote the right-hand-side of formula~\eqref{eq_Mf_new} by~$\mu$. We must show that $\mu = M_f$. Since $\deg\mu < \deg r = s$, it suffices to show that the two polynomials agree at~$s$ distinct points, namely the roots of~$f$ (hence also~$r$). Take any root $\alpha\in \CZ$, then by Lemma~\ref{lem_p_xi} and Eq.~\eqref{eq_xgcd} we may write
\[
\mu(\alpha) 
= P(\alpha)\cdot g(\alpha)
= m(\alpha)\cdot r'(\alpha)\cdot g(\alpha)
= m(\alpha)\cdot(1 - r\cdot h)(\alpha)
= m(\alpha).
\]
This shows that these two polynomials are identical.
\end{proof}

Let us now recall Guersenzvaig--Szechtman algorithm.

\begin{alg}
Given a monic polynomial~$f$, over a field of characteristic~$0$, this algorithm computes its square-free decomposition~\eqref{eq_sqfree_factorization}.
\begin{enumerate}
\renewcommand{\theenumi}{\alph{enumi}}
\item\label{step_a} Construct polynomials
\[
P := \frac{f'}{\gcd(f,f')},\qquad r := \frac{f}{\gcd(f,f')}
\]
and find polynomials $g,h$ satisfying Eq.~\eqref{eq_xgcd};
\item\label{step_b} build the polynomial~$M_f$ using either formula~\eqref{eq_Mf_original} or formula~\eqref{eq_Mf_new};
\item\label{step_c} for $k = 1, 2,\dotsc$ and as long as $\sum_{j\leq k} j\cdot \deg P_j < \deg f$, compute $k$-th square-free factor of~$f$:
\[
P_k := \gcd\bigl( M_f - k, r\bigr).
\]
\end{enumerate}
\end{alg}

\subsubsection*{Complexity analysis} As it is a common practice, let $M(n)$ denote the time complexity of computing a product of two polynomials of degree~$n$. For long polynomial multiplication $M(n) = O(n^2)$, while for fft-based multiplication $M(n) = O(n\cdot \lg n)$. (Here $\lg = \log_2$ stands for logarithm of base~$2$.) Recall that $s = \deg r$ and let $n := \deg f$. It is clear that $s = O(n)$. If one uses fast gcd, then step~\eqref{step_a} of the algorithm has time complexity $O\bigl(M(n)\lg n\bigr)$. The same applies to every iteration in step~\eqref{step_c}. The number of iterations equals $m = O(n^{0.5})$, the number of square-free factors of~$f$. Hence the complexity of step~\eqref{step_c} is $O\bigl(n^{0.5}\cdot M(n)\cdot \lg n\bigr) = O\bigl( n^{1.5}\cdot (\lg n)^2\bigr)$, when using fast polynomial multiplication and $O(n^{2.5}\lg n)$ for long multiplication. It remains to analyze the time complexity of step~\eqref{step_b}, the one for which we propose a new formula.

Observe that in order to construct~$M_f$ with formula~\eqref{eq_Mf_original}, one needs to evaluate the polynomial~$P$ on the matrix~$C_r$. Using Horner scheme, this requires~$s$ products of $s\times s$ square matrices. Thus, with standard matrix arithmetic, Eq.~\eqref{eq_Mf_original} needs $s + s^4$ scalar multiplications. Using fast matrix multiplication (see e.g. \cite{LeGall14, Zhdanovich12}), the asymptotic complexity (the actual number of scalar products is hard to count) of that formula is about $O(s^{3.37}) = O(n^{3.37})$. Therefore, if formula~\eqref{eq_Mf_original} is in use, step~\eqref{step_b} asymptotically dominates the running time of the whole algorithm.

On the contrary, formula~\eqref{eq_Mf_new} needs only $s+M(s)$ scalar products and so it has the asymptotic complexity of $O(s\lg s) = O(n\lg n)$, if one uses fast polynomial multiplication (respectively $O(s^2) = O(n^2)$ for long multiplication). Consequently, with formula~\eqref{eq_Mf_new}, step~\eqref{step_b} no longer dominates the time complexity of the algorithm.

In order to compare both formulas in real-life situations, they were implemented in computer algebra system Magma \cite{magma} and evaluated on random polynomials of different degrees. Table~\ref{tbl} presents the running times. The code can be downloaded from author's website  \url{http://z2.math.us.edu.pl/perry/papersen.html}

\subsubsection*{Conclusion} The modification of Guersenzvaig--Szechtman algorithm presented in this note leads to a procedure that is faster than the original one both theoretically and empirically.


\begin{table}
\caption{\label{tbl}Running times (in seconds) of Magma implementations of Eqs.~\eqref{eq_Mf_original} and~\eqref{eq_Mf_new}. For each degree, both formulas were evaluated on the same set of 10 randomly generated polynomials.}
\begin{center}
\begin{tabular}{r|r@{.}l|r@{.}l}
degree & \multicolumn{2}{c|}{Eq.~\eqref{eq_Mf_original}} & \multicolumn{2}{c}{Eq.~\eqref{eq_Mf_new}}\\\hline
10 & 0&010 & 0&000\\
20 & 0&020 & 0&010\\
50 & 0&170 & 0&140\\
100 & 1&140 & 0&600\\
200 & 10&690 & 3&170\\
500 & 330&900 & 32&230
\end{tabular}
\end{center}
\end{table}
\end{document}